\pdfoutput=1
\documentclass[twocolumn, 3p]{elsarticle}
\usepackage[utf8]{inputenc}

\usepackage{microtype}
\usepackage{amsfonts}
\usepackage{amsmath}
\usepackage{amsthm}
\usepackage{stmaryrd}
\usepackage{amssymb}
\SetSymbolFont{stmry}{bold}{U}{stmry}{m}{n}
\usepackage{dsfont}
\usepackage{xcolor}
\usepackage[linesnumbered,ruled,vlined]{algorithm2e}
\usepackage{hyperref}

\usepackage{microtype} 
\usepackage{ellipsis} 
\usepackage[abbreviations,british]{foreign} 

\newcommand{\logic}[1]{\mathsf{#1}}
\newcommand{\FO}{\logic{FO}}
\newcommand{\FOt}{\FO^{2}}

\newcommand{\oinvFOt}{{\lesseq}\text{-inv}~\FOt}

\newcommand{\complexityclass}[1]{\textsc{#1}} 

\newcommand{\PTime}{\complexityclass{PTime}} 
\newcommand{\NExpTime}{\complexityclass{NExpTime}} 
\newcommand{\coNExpTime}{\textrm{co}\complexityclass{NExpTime}} 
\newcommand{\coTwoNExpTime}{\textrm{co}\complexityclass{N2ExpTime}} 

\newcommand{\str}[1]{{\mathfrak{#1}}}

\renewcommand{\restriction}{\mathord{\upharpoonright}}
\newcommand{\restr}[2]{#1\restriction_{#2}} 

\newcommand{\relsymbol}[1]{\mathrm{#1}}
\newcommand{\relsymbolP}{\relsymbol{P}}

\newcommand{\domelem}[1]{\mathrm{#1}}                           
\newcommand{\domelemd}{\domelem{d}}                             
\newcommand{\domeleme}{\domelem{e}}                             

\newcommand{\var}[1]{\mathit{#1}}       
\newcommand{\varx}{\var{x}}             
\newcommand{\vary}{\var{y}}             
\newcommand{\vartuplex}{\vec{\varx}}    

\newcommand{\lesseq}{\preceq}

\usepackage{mathrsfs}
\newcommand{\AAA}{\mbox{\large \boldmath $\alpha$}}
\newcommand{\BBB}{\mbox{\large \boldmath $\beta$}}
\newcommand{\tp}[2]{{\rm tp}^{\str{#1}}({#2})}

\newtheorem{theorem}{Theorem}
\newtheorem{lemma}[theorem]{Lemma}

\newtheorem{corollary}[theorem]{Corollary}
\newtheorem{fact}[theorem]{Fact}

\SetCommentSty{mycommfont}
\SetKwInput{KwData}{Input}

\usepackage{balance}

\begin{document}

\begin{frontmatter}
  \title{Order-Invariance of Two-Variable Logic is $\coNExpTime$-complete}

  \author{Bartosz~Bednarczyk}
  \ead{bartosz.bednarczyk@cs.uni.wroc.pl}
  \address{Computational Logic Group, TU Dresden \& Institute of Computer Science, University of Wrocław}

  \begin{abstract}
    \noindent 
    We establish $\coNExpTime$-completeness of the problem of deciding order-invariance of a given two variable first-order formula, improving and significantly simplifying $\coTwoNExpTime$ bound by Zeume and Harwath. 
  \end{abstract}

  \begin{keyword}
    satisfiability \sep complexity \sep order-invariance 
    \sep classical decision problem \sep two-variable logic
  \end{keyword}
\end{frontmatter}


\section{Introduction}\label{sec:introduction}
The main goal of finite model theory is to understand formal languages describing finite structures: their complexity and their expressive power.
Such languages are ubiquitous in computer science, starting from descriptive complexity, where they are used to provide machine-independent characterisations of complexity classes, and ending up on database theory and knowledge-representation, where formal languages serve as fundamental querying formalism.
A classical idea in finite model theory is to employ invariantly-used relations, capturing data-independence principle in databases: it makes sense to give queries the ability to exploit the presence of the order in which the data is stored in the memory but at the same time we would like to make query results independent of this specific ordering.
It is not immediately clear that the use of invariantly-used linear order in first-order logic ($\FO$) allow us to gain anything expressivity-wise. 
And as soon as we stick to the arbitrary (\ie not necessarily finite) structures it does not, which is a direct consequence of $\FO$ having the Craig Interpolation Property.
However, as it was first shown by Gurevich~\cite[Thm.~5.3]{Libkin04}, the claim holds true over finite structures: order-invariant $\FO$ is more expressive than plain~$\FO$.

Unfortunately, order-invariant $\FO$ is poorly understood. 
As stated in~\cite{BarceloL16} one of the reasons why the progress in understanding order-invariance is rather slow is the lack of logical toolkit.
The classical model-theoretic methods based on types were proposed only recently~\cite{BarceloL16}, and the order-invariant $\FO$ is not even a logic in the classical sense: its syntax is undecidable. 
Moreover, the availability of locality-based methods is limited: order-invariant $\FO$ is known to be Gaifman-local~\cite[Thm. 2]{GroheS00} but the status of an analogous of Hanf-locality from $\FO$ is open.
This suggest that a good way to understand order-invariant $\FO$ is to first look at its fragments, \eg the fragments with a limited~number~of~variables.

\paragraph*{Our contribution}\label{subsec:our-contribution}
We continue the research on the two-variable fragment $(\FOt)$ of order-invariant $\FO$, initiated in~\cite{ZeumeH16}.
It was shown, in contrast to the full $\FO$, that checking order-invariance of an input $\FOt$ is decidable in $\coTwoNExpTime$, see:~\cite[Thm.~12]{ZeumeH16}.
We provide a tight bound for the mentioned problem, showing that deciding order-invariance for $\FOt$ is $\coNExpTime$-complete. 
Our proof method relies on establishing exponential-size counter-example for order-invariance and is \emph{surprisingly easy}. 


\section{Preliminaries}\label{sec:preliminaries}
\noindent We employ standard terminology from finite model theory, assuming that the reader is familiar with the syntax and the semantics of first-order logic ($\FO$)~\cite[Sec.~2.1]{Libkin04}, basics on computability and complexity~\cite[Secs.~2.2--2.3]{Libkin04}, and order-invariant queries~\cite[Secs.~5.1--5.2]{Libkin04}.
In what follows, $\FOt$ denotes the set of all $\FO$ sentences employing only the variables~$\varx, \vary$.

\paragraph*{Structures}\label{para:intro-structures}
Structures are denoted with fraktur letters $\str{A}, \str{B}$ and their domains are denoted with the corresponding Roman letters $A, B$.
We assume that structures have non-empty, \emph{finite} domains, and are over some purely-relational vocabulary.
For paper-specific reasons we employ three distinguished symbols $\lesseq, \lesseq_0, \lesseq_1$ that are interpreted as linear orders (\ie a reflexive, antisymmetric, transitive and total relation).
We write $\FO[\Theta]$ for $\Theta \subseteq \{ \lesseq, \lesseq_0, \lesseq_1 \}$ to indicate that only distinguished symbols from $\Theta$ may appear in $\varphi$, and $\varphi[{\lesseq}{/}{\lesseq_i}]$ to denote the formula obtained from $\varphi$ by replacing each occurrence of $\lesseq$ with~$\lesseq_i$.
We write~$\varphi(\vartuplex)$ to indicate that all free variables of~$\varphi$ are in~$\vartuplex$.
A~sentence is a formula without free variables.
With $\restr{\str{A}}{S}$ we denote the substructure of the structure $\str{A}$ restricted to the set $S \subseteq A$.

\paragraph*{Types}\label{para:intro-types}
An (atomic) \emph{$1$-type} over $\tau$ is a maximal satisfiable set of atoms or negated atoms from $\tau$ with a free variable~$\varx$. 
Similarly, an (atomic) \emph{$2$-type} over $\tau $ is a maximal satisfiable set of atoms or negated atoms with free variables~$\varx,\vary$. 
Note that the total number of atomic $1$- and $2$-types over $\tau$ is bounded exponentially in~$|\tau|$. 
We often identify a type with the conjunction of all its elements. 
The set of $1$-types and $2$-types over the signature consisting of symbols appearing in $\varphi$ is denoted with $\AAA_\varphi$ and $\BBB_\varphi$.
Given a structure $\str{A}$ and its element~$\domelemd \in A$ we say that $\domelemd$ \emph{realises} a $1$-type $\alpha$ if $\alpha$ is the unique $1$-type such that $\str{A} \models \alpha[d]$.
Similarly, for distinct $\domelemd, \domeleme \in A$, we denote by $\tp{A}{\domelemd, \domeleme}$ the unique $2$-type \emph{realised} by the pair $(\domelemd, \domeleme)$, \ie the $2$-type $\beta$ such that~$\str{A} \models \beta[\domelemd, \domeleme]$.

\paragraph*{Decision problems}\label{para:intro-decision-problems}
The \emph{finite satisfiability} (resp. \emph{validity}) \emph{problem} for a logic $\logic{L}$ is the problem of deciding whether an input sentence $\varphi$ from $\logic{L}$ is satisfied in some (resp. all) finite structures.
Recall that the finite satisfiability and validity for $\FO$ are undecidable~\cite{turing1938computable,church1936note}, while for $\FOt$ are respectively $\NExpTime$-complete and $\coNExpTime$-complete~\cite[Thm.~5.3]{GradelKV97}\cite[Thm.~3]{Furer83}.
Note that $\varphi$ is finitely valid iff $\neg \varphi$ is finitely unsatisfiable.

\paragraph*{Order-invariance}\label{para:intro-order-inv}
A sentence $\varphi$ is \emph{order-invariant} (or: $\lesseq$-invariant) if for all finite $\tau$-structures~$\str{A}$ and all (linear-order!) interpretations $\lesseq^{\str{A}}, (\lesseq^{\str{A}})'$ of $\lesseq$ over $A$ we have that $(\str{A}, \lesseq^{\str{A}}) \models \varphi$ iff $(\str{A}, (\lesseq^{\str{A}})') \models \varphi$.
With $\oinvFOt$ we denote the set of all $\lesseq$-invariant $\FO^2$ sentences.
Note that $\varphi$ is \emph{not} order-invariant if there is an $\str{A}$ and two linear orders $\lesseq_0, \lesseq_1$ on $A$ such that $(\str{A}, \lesseq_0) \models \varphi$ and $(\str{A}, \lesseq_1) \not\models \varphi$. 
Deciding if an $\FO$ sentence is order-invariant is undecidable~\cite[Ex.~9.3]{Libkin04}. 
Checking order-invariance for $\FOt$ formulae was shown to be in $\coTwoNExpTime$ in~\cite[Thm.~12]{ZeumeH16}.\footnote{The authors of~\cite{ZeumeH16} incorrectly stated the complexity in their Thm. 12, mistaking ``invariance'' with ``non-invariance''.}


\section{Deciding Order-Invariance}\label{sec:complexity}
We study the complexity of the problem of deciding if an input formula $\varphi[\lesseq] \in \FO^2$ is order-invariant.
We start from the lower bound first. We consider the following program, inspired by~\cite[Slide 9]{Schweikardt13Slides}.

\begin{algorithm}[h]
  \DontPrintSemicolon
  \KwData{An $\FOt[\emptyset]$-formula $\varphi$.}
  \caption{From validity to $\lesseq$-invariance}\label{algo:a-reduction-from-validity}

  If $\neg\varphi$ has a model with a single-element domain, \textbf{return} \texttt{False}. \tcp*{A corner case}

  Let $\psi_{\lesseq} := \exists{\varx} \left( \relsymbolP(\varx) \land \forall{\vary} (\vary \lesseq \varx) \right)$ for $\relsymbolP \not\in \tau$.\tcp*{not $\lesseq$-inv. on structs. with ${\geq} 2$ elements!}

  \textbf{Return} \texttt{True} if $(\neg \varphi) \to \psi_\lesseq$ is $\lesseq$-invariant and \texttt{False} otherwise. \tcp*{Actual reduction.}
\end{algorithm}

The above program Turing-reduces finite $\FOt$-validity to testing order-invariance of $\FOt$-sentences.
It is straightforward to see that the presented reduction is correct, so we we left the details for the reader.
From the complexity bounds on the finite validity problem for $\FOt$~\cite[Thm.~3]{Furer83} we conclude:

\begin{corollary}\label{lemma:correctness-of-the-reduction}
    Procedure~\ref{algo:a-reduction-from-validity} returns \texttt{True} iff its input is finitely valid. 
    Hence, testing whether an $\FOt$ formula is order-invariant is $\coNExpTime$-hard.
\end{corollary}

\noindent The upper bound relies on the following fact that follows directly from the~definition~of~order-invariance.
\begin{fact}\label{fact:from-inv-to-sat}
An $\FOt[\lesseq]$ formula $\varphi$ is \emph{not} order-invariant iff the $\FOt[\lesseq_0, \lesseq_1 ]$ formula $\varphi[{\lesseq}{/}{\lesseq_0}] \land \neg \varphi[{\lesseq}{/}{\lesseq_1}]$ is finitely-satisfiable over structures interpreting $\lesseq_0$ and $\lesseq_1$ as linear orders over the domain.
\end{fact}

Let $\FOt_-[\lesseq_0, \lesseq_1]$ be composed of all sentences of the shape $\varphi[{\lesseq}{/}{\lesseq_0}] \land \neg \varphi[{\lesseq}{/}{\lesseq_1}]$ with $\varphi \in \FOt[\lesseq]$. 
We stress that we always assume that $\lesseq$ symbols are interpreted as linear orders over the domain.
To simplify the reasoning about such formulae, we first reduce them to Scott-like \emph{normal forms}, \cf~\cite[\textsection4]{GradelKV97}, \cite[Sec.~3.1]{Otto01}.
By applying \cite[Lemma 1]{ZeumeH16} to $\varphi[{\lesseq}{/}{\lesseq_0}]$ and $\neg \varphi[{\lesseq}{/}{\lesseq_1}]$, and taking their conjunction, we~infer:
%
\begin{corollary}\label{lemma:normalform}
  For any $\FOt_-[\lesseq_0, \lesseq_1]$ formula there is an equi-satisfiable, linear-time computable formula (over an extended signature) having the form: 
  \[
    \bigwedge_{i=0}^{1} \left( \forall{\varx}\forall{\vary} \; \chi_i(\varx, \vary) \land \bigwedge\limits_{j=1}^{m_i} \forall{\varx}\exists{\vary} \ \gamma_i^j(\varx, \vary) \right),
  \]
  where the decorated  $\chi$ and $\gamma$ are quantifier-free and the symbols $\lesseq_i$ do not appear in $\chi_{1-i}$ and $\gamma_{1-i}^{j}$. 
\end{corollary}

Given a model $\str{A} \models \varphi$ of $\varphi$ in normal form and elements $\domelemd, \domeleme \in A$ witnessing $\str{A} \models \gamma_i^j(\domelemd, \domeleme)$, we call $\domeleme$ a $\gamma_i^j$-\emph{witness} for $\domelemd$ (or simply a witness).

The core of the paper is the following small model theorem, with a proof relying on the circular witnessing scheme by Grädel-Kolaitis-Vardi~\cite[Thm.~4.3]{GradelKV97}. 

\begin{lemma}\label{lemma:small-model-property}
Any finitely satisfiable $\varphi \ {\in}\ \FOt_-[\lesseq_0, \lesseq_1]$ has a model with $\mathcal{O}(|\varphi|^3 \cdot 2^{|\varphi|})$ elements.
\end{lemma}
\begin{proof}
W.l.o.g. we assume that $\varphi$ is in the normal form from Corollary~\ref{lemma:normalform} and put $M := \max{(m_0,m_1)}$. 
Let $\str{A} \models \varphi$. If $|A| \leq 224\ |\varphi|^3 \cdot 2^{|\varphi|}$ then we are done, so assume otherwise.
We are going to construct a new model $\str{B}$ having the domain $W_0 \cup W_1 \cup W_2 \cup W_3$, where the sets $W_i$ are constructed below.

Call $1$-type \emph{rare} if it has at most $32M$ realisations in~$\str{A}$.
Let $S$ be composed of all elements of~$\str{A}$ of rare $1$-types, and of the $8M$ minimal and $8M$ maximal (w.r.t. each $\lesseq_0^{\str{A}}$, $\lesseq_1^{\str{A}}$) realisations of each non-rare $1$-type in $\str{A}$. 
We make $W_0$ to be composed of all elements realising rare-types, as well as $M$ minimal and $M$ maximal (w.r.t. $\lesseq_0^{\str{A}}$ and $\lesseq_1^{\str{A}}$) realisations of each non-rare $1$-type in $\str{A}$. 
Put the rest of elements of $S$ to $W_1$.
We have $|W_0 \cup W_1| \leq 32M \cdot |\AAA_{\varphi}|$.

The idea behind $W_0$ is that this set contains ``dangerous'' elements, \ie the ones for which $\restr{\str{A}}{W_0}$ may be uniquely determined by $\varphi$.
Elements from $W_1$ will help to restore the satisfaction of $\forall\exists$ conjuncts.

Finally, we close $W_0 \cup W_1$ twice under witnesses.
More precisely, let $W_2$ be any $\subseteq$-minimal subset of $A$ so that all elements from $W_0 \cup W_1$ have all the required $\gamma_i^j$-witnesses in $W_0 \cup W_1 \cup W_2$.
  Similarly, we define $W_3$ to be any $\subseteq$-minimal subset of $A$ so that all elements from $W_0 \cup W_1 \cup W_2$ have all the required $\gamma_i^j$-witnesses in $W_0 \cup W_1 \cup W_2 \cup W_3$.
Observe that $|W_2| \leq 2M |W_0 \cup W_1| \leq 2M \cdot 32 M |\AAA_{\varphi}| = 64 M^2 |\AAA_{\varphi}|$
  and $|W_3| \leq 2M |W_2| \leq 128 M^3 |\AAA_{\varphi}|$ hold.

Let $\str{B} := \restr{\str{A}}{W_0 \cup W_1 \cup W_2 \cup W_3}$.
We see that:
\begin{align*}
|B| \leq |W_0 \cup W_1| + |W_2| + |W_3| \leq 
(32M + 64M^2 +\\ + 128M^3) |\AAA_\varphi| \leq
224 M^3 |\AAA_\varphi| \leq 224\ |\varphi|^3 \cdot 2^{|\varphi|}.
\end{align*}

Universal formulae are preserved under substructures, thus $\lesseq_1^{\str{B}}, \lesseq_2^{\str{B}}$ are linear orders over $B$ and $\str{B}$ satisfies $\forall\forall$-conjuncts of $\varphi$.
Hence, the only reason for $\str{B}$ to not be a model of $\varphi$ is the lack of required $\gamma_i^j$-witnesses for elements from $W_3$. 
We fix this by reinterpreting relations between $W_3$ and $W_1$.

Before we start, we are going to collect per each non-rare $1$-type $\alpha$, pairwise-disjoint sets of $M$ minimal and $M$ maximal (w.r.t. each $\lesseq_0^{\str{A}}$, $\lesseq_1^{\str{A}}$) realisations of $\alpha$ from $W_1$.
Formally:
Fix a non-rare $\alpha$.
Let $V_\alpha^{0}$ be composed of the first $M$ $\lesseq_0$-minimal elements from $\restr{\str{A}}{W_1}$.
Next, let $V_\alpha^{1}$ be composed of the last $M$ $\lesseq_0$-maximal elements from $\restr{\str{A}}{W_1 \setminus V_\alpha^{0}}$, 
Similarly, let $V_\alpha^{2}$ be composed of the first $M$ $\lesseq_1$-minimal elements from $\restr{\str{A}}{W_1 \setminus (V_\alpha^{0} \cup V_\alpha^{1})}$.
Finally let $V_\alpha^{3}$ be composed of the last $M$ $\lesseq_1$-maximal elements~from~$\restr{\str{A}}{W_1 \setminus (V_\alpha^{0} \cup V_\alpha^{1} \cup V_\alpha^{2})}$.
Put $V_\alpha := \bigcup_{k=0}^{3} V_\alpha^{k}$. 
Notice that all the components of $V_\alpha$ are pairwise disjoint (by construction), and they are well-defined since we included sufficiently many elements~in~$W_1$.

Going back to the proof, we fix any element $\domelemd$ from $W_3$ that violate some of $\forall\exists$-conjuncts of $\varphi$.
Next, fix any $\forall\exists$-conjunct~$\psi := \forall{\varx}\exists{\vary}\ \gamma_i^j(x,y)$, whose satisfaction is violated by $\domelemd$.
Since $\str{A} \models \varphi$ we know that there is an element $\domeleme \in A$ such that $\domeleme$ is a $\gamma_i^j$-witness for $\domelemd$ and $\gamma_i^j$ in $\str{A}$ and let $\alpha$ be the $1$-type of $\domeleme$ in $\str{A}$. 
Observe that $\alpha$ is not rare (otherwise $\domeleme \in W_0$, and hence $\domeleme \in B$), and $\domelemd \neq \domeleme$.
Moreover either $\domeleme \lesseq_i^{\str{A}} \domelemd$ or $\domelemd \lesseq_i^{\str{A}} \domeleme$ holds. 
Thus, we take $V_{\alpha}^{2i+k}$ (where $k$ equals $0$ if $\domeleme \lesseq_i^{\str{A}} \domelemd$ and $1$ otherwise) to be the corresponding set of $M$ minimal/maximal $\lesseq_i$ realisations of $\alpha$ in the same direction to~$\domelemd$~as~$\domeleme$~is.

Now it suffices to take the $j$-th element $\domeleme_j$ from $V_{\alpha}^{2i{+}k}$ and change the binary relations between $\domelemd$ and $\domeleme_j$ in $\str{B}$ so that the equality holds $\tp{A}{\domelemd, \domeleme} = \tp{B}{\domelemd, \domeleme_j}$ holds (which can be done as $\domeleme$ and $\domeleme_j$ have equal $1$-types).
We repeat the process for all remaining $\gamma_i^j$ formulae violated by $\domelemd$.
We stress that it is not a coincidence that we use the $j$-th element $\domeleme_j$ from the corresponding set $V_{\alpha}^{2i{+}k}$ to be a fresh $\gamma_i^j$-witness for $\domelemd$: 
this guarantees that we never redefine connection between $\domelemd$ and some element twice.

Observe that all elements from $B$ that had $\gamma_i^j$-witnesses before our redefinition of certain $2$-types, still do have them (as we did do not touch $2$-types between them and their witnesses), $\str{B}$ still satisfies the $\forall\forall$-component of~$\varphi$ (since the modified $2$-type does not violate $\varphi$ in $\str{B}$ it does not violate $\varphi$ in $\str{B}$) and $\domelemd$ has all required witnesses. 
By repeating the strategy for all the other elements from $W_3$ violating $\varphi$, we obtain a desired ``small'' model of $\varphi$. 
\qedhere
\end{proof}

Lemma~\ref{lemma:small-model-property} yields an $\NExpTime$ algorithm for deciding satisfiability of $\FOt_-[\lesseq_0, \lesseq_1]$ formulae: convert an input into normal form, guess its exponential size model and verify the modelhood with a standard model-checking algorithm (in $\PTime$~\cite[Prop. 4.1]{GradelO99}). 
After applying Proposition~\ref{fact:from-inv-to-sat} we~conclude:
\begin{theorem}\label{thm:complexity}
Checking if an $\FOt[\lesseq]$-formula is $\lesseq$-invariant is $\coNExpTime$-complete.
\end{theorem}


\section*{Acknowledgements}
This work was supported by the ERC through 
the Consolidator Grant No.~771779 (\href{https://iccl.inf.tu-dresden.de/web/DeciGUT/en}{DeciGUT}).

\noindent I thank Antti Kuusisto and Anna Karykowska for inspiring discussions, as well as Reijo Jaakkola, Julien Grange and Emanuel Kieroński for proofreading.


\balance
\bibliographystyle{plain}
\bibliography{references}
\end{document}